%% file: main.tex
\documentclass[11pt]{article}
\usepackage{mathrsfs}
\usepackage{amsmath}
\usepackage[colorlinks=true, allcolors=blue, urlcolor=black]{hyperref}

\makeatletter
\newcommand{\removelatexerror}{\let\@latex@error\@gobble}
\makeatother
\usepackage{amsmath}
\usepackage{amsthm}
\usepackage{amsfonts}
\usepackage{amssymb}
\usepackage{graphicx}
\usepackage{dsfont}
\usepackage{csquotes}
\usepackage{mathtools}
\usepackage{cleveref}
\usepackage{pgfplots}
\pgfplotsset{compat=1.18}
\usepackage{bbm}
\usepackage{stmaryrd}

\usepackage{complexity}
\usepackage{enumitem}
\usepackage{tikz}
\usetikzlibrary{positioning}
\usepackage{thm-restate}
\usepackage{float}
\usetikzlibrary{calc}
\usepackage{thm-restate}
\usetikzlibrary{shapes,decorations.markings}
\usepackage{caption}
\usepackage{subcaption}
\usepackage{comment}

\usepackage{lipsum}

\newtheorem{theorem}{Theorem}[section]

\newtheorem{definition}[theorem]{Definition}
\newtheorem{claim}[theorem]{Claim}
\newtheorem{lemma}[theorem]{Lemma}
\newtheorem{proposition}[theorem]{Proposition}

\newtheorem{fact}[theorem]{Fact}

\renewcommand{\E}{\mathbb{E}}
\newcommand{\Var}{\textnormal{Var}}
\newcommand{\Varsub}[1]{\underset{#1}{\Var}}

\newcommand{\Esub}[1]{\underset{#1}{\E}}
\newcommand{\GapMAJ}{\mathsf{GapMAJ}}

\newcommand{\VarObsRound}[0]{\overline{\Phi}_{\textnormal{obs}}}
\newcommand{\VarBobRound}[0]{\overline{\Phi}_{\textnormal{Bob}}}

\usepackage[english]{babel}
\usepackage[T1]{fontenc}
\usepackage{caption}

\usepackage[ruled,vlined,linesnumbered]{algorithm2e}
\crefname{algocf}{Algorithm}{Algorithms}

\usepackage[margin=1in]{geometry}

\title{An Optimal Gap-Majority Lemma for \\ Bounded-Round Communication}
\author{ Pachara Sawettamalya\footnote{Department of Computer Science, Princeton University. Supported by Huacheng Yu's NSF CAREER award CCF-233994. \ \href{mailto:ps3122@princeton.edu}{\url{ps3122@princeton.edu}}}}
\date{\today}

\usepackage[backend=biber, style=alphabetic, doi=false, url=false, isbn=false, minalphanames = 3, maxalphanames = 4, maxbibnames=10]{biblatex}
\begin{document}
\maketitle

\vspace{-7mm}

\input{0-abstract}

\pagenumbering{roman}
\pagenumbering{arabic}

\input{1-intro}

\input{2-prelim}

\input{3-overview}

\input{4-variance}

\input{5-variance-1way}

\input{6-discussion}

\printbibliography

\end{document}

%% file: 0-abstract.tex
\begin{abstract}

Recent work of Sawettamalya and Yu (FOCS 2026) initiated the study of \emph{gap-majority} composition in randomized communication. They showed that if computing a $\{\pm 1\}$-valued function $f$ in $r$ rounds requires $I$ bits of information, then distinguishing whether $\sum_{i=1}^n f(X_i,Y_i)$ is at least $0.01\sqrt{n}$ or at most $-0.01\sqrt{n}$ in $r$ rounds requires $n \cdot (I-O(\log n))$ bits of information. They also conjectured that the $O(\log n)$-additive ``loss'' can be further reduced to $O(1)$.

In this note, we show that a minor adaptation of their approach, together with a more refined analysis, achieves the conjectured $O(1)$ loss, establishing a gap-majority lemma for bounded-round communication that is optimal up to constant factors.

\end{abstract}

%% file: 1-intro.tex
\section{Introduction}
\label{sec:intro}

The study of the complexity of composed functions is central to computational complexity. It asks how difficult it is to compute a composed function $g \circ f^n$ relative to computing the base function $f$ and the outer function $g$.

Within communication complexity \cite{Yao79}, the study of this question has been somewhat limited. Early work on \emph{information complexity} showed that computing $g \circ f^n$ requires $\Omega(n)$ times the information needed to compute $f$ when the outer gadget $g$ is either identity or XOR \cite{BYJKS04, BR11, JPY12, BBCR13, BRWY13, Braverman15, Yu22, IR24b, SawettamalyaY25}, meaning that one cannot hope to save information by avoiding solving all $n$ instances of $f$. Progress beyond these ``gadgets'', however, remained stagnant for decades as no other natural, explicit gadget had been shown to admit such property. Only very recently did Sawettamalya and Yu \cite{SY26} establish similar results for composition with the \emph{gap-majority} function, denoted $\GapMAJ$, defined as
\begin{align*}
\GapMAJ_n(z_1,\ldots,z_n) :=
\begin{cases}
+1 & \text{if } \sum_{i=1}^{n} z_i \geq 0.01 \sqrt{n}, \\
-1 & \text{if } \sum_{i=1}^{n} z_i \leq -0.01 \sqrt{n}.
\end{cases}
\end{align*}
In fact, \cite{SY26} proved two versions of the \emph{gap-majority} lemma summarized as follows.
\begin{enumerate}
    \item If computing $f$ requires $I$ bits of information, then computing $\GapMAJ_n \circ f^n$ requires $n \cdot (I-O(1))$ bits of information.
    \item For any $r$, if computing $f$ within $r$ rounds requires $I$ bits of information, then computing $\GapMAJ_n \circ f^n$ within $r$ rounds requires $n \cdot (I-O(\log n))$ bits of information.
\end{enumerate}
The first result gives a stronger quantitative bound ($O(1)$ loss versus $O(\log n)$), but the retaining the number of communication rounds in the second is a natural feature which appear useful for broad applications. Observing this discrepancy, they asked whether ``the best of both worlds'' could be achieved.

In this note, we record that little adjustments to the approach of \cite{SY26} combined a more precise analysis already achieves this goal. We summarize the our result as follows.

\begin{theorem}[Gap-majority lemma for bounded-round communication]
Let $f: \mathcal{X} \times \mathcal{Y} \rightarrow \{\pm 1\}$ and $\mu$ be an input distribution to $f$ such that $\E_\mu[f(x,y)] = 0$. Let $\pi$ be an $r$-round protocol that computes $\GapMAJ_n \circ f^n$ correctly with probability $0.99$ over $\mu^n$. Then, there exists an $r$-round protocol $\eta$ that computes $f$ correctly with probability $0.501$ over $\mu$ and $\IC(\eta) \le \frac{\IC(\pi)}{n} + O(1).$
\label{thm:round_gap_majority}
\end{theorem}

\subsection{Concurrent work and disclaimer}

A recent manuscript by Wang \cite{Wang26} also establishes the same result, among others. The proofs appear similar in some respects, although they use different formulations.

This write-up was prepared under time constraints with assistance from GPT 6.1 Sol and does not yet reflect the author's preferred presentation. Further details on AI usage are discussed in \Cref{AI}.

%% file: 2-prelim.tex
\section{Preliminaries}
\label{sec:prelim}

We refer the reader to the preliminaries of \cite{SY26} for which we largely follow. Here, we highlight some key borrowed notions, as well as facts and terms specific to this paper.

\paragraph{Notation.} For a real-valued base function $f$, we define its \emph{$n$-fold summation} as $f^{+n}(X_1,\ldots,X_n) := f(X_1) + \ldots + f(X_n)$. We occasionally write $f^{+n}$ simply as $f^{+}$, omitting the number of summands when it is clear from context.

Following standard convention in probability, we use uppercase letters to denote random variables and lowercase letters for their realizations. We write $\E[A]$ and $\Var[A]$ for the expectation and variance of $A$, respectively, and use subscripts to indicate the variable of integration, e.g., $\E_A[f(A)]$. When clear from context, we omit the subscript.

\paragraph{Communication protocols with bounded rounds.} We view an $r$-round communication protocol $\pi$ as a joint distribution over random variables $(X,Y,M,A,R)$. Here, $(X,Y)\sim\mu$ are the players' inputs, $M=(M_1,\ldots,M_r)$ is the sequence of alternating messages, $A$ is Bob's answer, and $R$ is the public randomness. We assume that $r$ is odd, with Alice sending the first and last messages.

\paragraph{Smoothed Protocols.} Throughout this note, we take $\mathcal{A}=\{\pm1\}$ to be the output alphabet for the gap-majority composition. We assume that the protocol is \emph{smooth} in the following sense: for $(X,Y,M,R)$, Bob must answer $A=+1$ and $A=-1$ with probability at least $0.001$. This feat can be achieved by running an original protocol as is, and with probability $0.001$ Bob flips the answer $A$. It turns out that this feature is needed for technical reasons, while only sacrificing $0.001$ in the error rates which is easily affordable.

\paragraph{Information cost.} The information cost of an $r$-round protocol $\pi=(X,Y,M,A,R)$ is
\begin{align*}
\IC(\pi) := I(M:X \mid YR) + I(M:Y \mid XR),
\end{align*}
where $I(\cdot:\cdot)$ denotes mutual information.

\begin{proposition}[Rectangle property; Section 3.4 of \cite{Yu22}]
Let $\pi=(X,Y,M,A,R)$ be an $r$-round protocol with $n$-coordinate inputs $(X,Y)\sim\mu^n$. For any partition $[n]=P\cup Q$, we have $X_P \perp Y_Q \mid MRAX_QY_P$.
\label{prop:rectangle}
\end{proposition}

We will also use the following facts.

\begin{fact}
Let $X$ be a random variable and $t \in \mathbb{R}$. Then, we have
$\E[(X-t)^2] = \Var[X] + \bigl(\E[X]-t\bigr)^2.$
\label{fct:var_bound}
\end{fact}

\begin{lemma}
Let $U$ be an integer-valued random variable. There exist universal constants $c_0,c_1$ such that
$
H(U) \leq c_0 + c_1 \cdot \E[U^2].
$
\label{lem:entropy_bound}
\end{lemma}

\begin{proof}
For $i\in\mathbb{Z}$, let $p_i=\Pr[U=i]$. Using $\log x\leq x$ with $x = e^{-|i|}/p_i$ and the fact that $|i| \leq i^2$, we obtain $p_i \log \frac{1}{p_i} \leq e^{-|i|} + p_i \cdot |i| \leq e^{-|i|} + p_i \cdot i^2.$ 
Summing over $i$, we have
\begin{align*}
\sum_{i \in \mathbb{Z}} p_i \log \frac{1}{p_i} \leq \sum_{i\in\mathbb{Z}}e^{-|i|}
+\sum_{i\in\mathbb{Z}}p_i \cdot i^2 
= O(1)+\E[U^2].
\end{align*}
Rescaling the logarithms concludes the proof.
\end{proof}

%% file: 3-overview.tex
\section{Overview}

We follow the proof of result 2 in \cite{SY26} which uses \emph{variance} as a proxy for the error rate. Recall the following definitions used in their work.

\begin{definition}[Definition 3.4 of \cite{SY26}]
\label{def:1-way-var}
Let $\pi=(X,Y,M,A,R)$ be an $r$-round protocol, and let $f$ be a function of the input pair $(X,Y)$. Define the \emph{observer's variance} of $\pi$ with respect to $f$ by
\begin{align*}
\VarObsRound(\pi @ f) := \Esub{MRA}\,\Varsub{XY}[f(X,Y) \mid MRA]
\end{align*}
and \emph{Bob's variance} of $\pi$ with respect to $f$ by
\begin{align*}
\VarBobRound(\pi @ f) := \Esub{MRAY}\,\Varsub{X}[f(X,Y) \mid MRAY].
\end{align*}
\end{definition}

Their proof proceeds in three steps. First, they show that any protocol $\pi$ that successfully computes $\GapMAJ_n \circ f^n$ must have low observer variance.

\begin{lemma}[\cite{BravermanGW20, SY26}]
Let $f: \mathcal{X} \times \mathcal{Y} \to \{\pm1\}$, and let $\mu$ be an input distribution for $f$ satisfying $\E_\mu[f(x,y)]=0$. Suppose $\pi$ is an $r$-round protocol that computes $\GapMAJ_n \circ f^n$ with success probability $0.99$ over $\mu^n$. Then $\VarObsRound(\pi @ f^{+n}) \leq 0.99n$.
\label{lem:gapmaj_then_low_variance_1way}
\end{lemma}

Second, they convert $\pi$ into a new protocol $\eta$ with proportionally smaller variance and information cost. This is the step where their additive loss of $O(\log n)$ arises. Here, we improve this step so that the loss is only $O(1)$.

\begin{lemma}
Let $f: \mathcal{X} \times \mathcal{Y} \rightarrow \{\pm1\}$ be an arbitrary function, and let $\mu$ be any input distribution over $\mathcal{X} \times \mathcal{Y}$ such that $\E_\mu[f(x,y)] = 0$. Suppose $\pi$ is an $r$-round protocol over $\mu^n$. Then, for any $\Delta \in(0,1)$, there exists an $r$-round protocol $\eta$ for computing $f$ over $\mu$ such that
\begin{align*}
\IC(\eta) \leq \frac{\IC(\pi)}{n}+O(1/\Delta)
\qquad \text{and} \qquad
\VarBobRound(\eta @ f) \leq \frac{\VarObsRound(\pi @ f^{+n})}{n}+\Delta.
\end{align*}
\label{lem:main_lemma_oneway}
\end{lemma}

Third and finally, we take $\Delta$ to be a sufficiently small positive constant and have Bob output the more likely value of $f$ given his view in $\eta$. This yields the desired information cost and error probability for computing $f$. For a more detailed outline, we refer the reader to Section 3 of \cite{SY26}.

\paragraph{Setting up the proof of \Cref{lem:main_lemma_oneway}.}

The remainder of this paper is devoted to proving \Cref{lem:main_lemma_oneway}. Let $f: \mathcal{X} \times \mathcal{Y} \rightarrow \{\pm1\}$, and let $\mu$ be an input distribution. Let $\pi=(X,Y,M,A,R)$ be an $r$-round protocol over the product distribution $\mu^n$, where $X=(X_1,\ldots,X_n)$ and $Y=(Y_1,\ldots,Y_n)$, with the pairs $(X_i,Y_i)$ drawn independently from $\mu$. Write $M=(M_1,\ldots,M_r)$ for the sequence of messages. We assume that $r$ is odd, so Alice sends the last message $M_r$ and Bob outputs the answer $A\in\{\pm1\}$. Recall also the assumption that $\pi$ is \emph{smooth}: conditional on any realization of $(X,Y,M,R)$, each answer $A\in\{\pm1\}$ occurs with probability at least $0.001$. Our goal is to construct an $r$-round protocol $\eta$ with the stated bounds on information cost and variance.

Following the strategy of \cite{SY26}, we explicitly construct $n$ candidate protocols $(\pi^1,\ldots,\pi^n)$, each with input distribution $\mu$, satisfying
\begin{align*}
\sum_{i\in[n]} \IC(\pi^i)
&\leq \IC(\pi)+O(n/\Delta) \qquad \text{and} \qquad
\sum_{i\in[n]} \VarBobRound(\pi^i @ f)
\leq \VarObsRound(\pi @ f^{+n})+n\Delta .    
\end{align*}
The protocol $\eta$ then samples $i\in[n]$ uniformly using public randomness and runs $\pi^i$.

We establish the bound on the sum of variances in \Cref{sec:var} and the bound on the sum of information costs in \Cref{sec:info}, completing the proof.

%% file: 4-variance.tex
\section{Decomposition of Variances}
\label{sec:var}

In this section, we show that the variances admit a linear decomposition up to a small ``constant'' additive loss. We follow the \emph{protocol decomposition} strategy of \cite{SY26}. Given an $r$-round protocol $\pi$, we construct two smaller protocols as follows. Recall that loss parameter $\Delta \in(0,1)$. Going forward, let $\delta \in (0,1)$ be such that $\Delta = \delta^2$.

\begin{figure}[H]
    \centering
    \begin{tabular}{|p{15.8cm}|}
    \hline ~\\
    \multicolumn{1}{|c|}{\textbf{A protocol $\pi^{<n}$ over the input distribution $(X',Y')\sim\mu^{n-1}$}} \\
    \\ \hline
    \vspace{-2.85mm}
    \begin{enumerate}
        \item Alice pretends that her input consists of $n$ coordinates and writes it as $X=(X_1,\ldots,X_n)$.
        \vspace{-5.8mm}
        \item Bob pretends that his input consists of $n$ coordinates and writes it as $Y=(Y_1,\ldots,Y_n)$.
        \vspace{-1mm}
        \item Alice embeds $X'$ into $X_{<n}$, and Bob embeds $Y'$ into $Y_{<n}$.
        \vspace{-1mm}
        \item Using public randomness, the players draw $Y_n$ from its marginal under $\mu$. Alice then privately draws $X_n$ from its conditional distribution given $Y_n$.
        \vspace{-1mm}
        \item Alice and Bob simulate $\pi$ on $(X,Y)$, generating the messages $M=(M_1,\ldots,M_r)$. Alice augments her final message as described below.
        \vspace{-1mm}
        \item For each $a\in \{\pm1\}$, Alice computes $$\hat{Z}_n^a:=\Esub{Y{<n}}\left[f^+(X_{<n},Y_{<n})\mid MRX_{<n}Y_n,A=a\right].$$ She rounds each value to the nearest multiple of $\delta$ to obtain $Z_n^a.$
        \item Alice appends $Z_n:=(Z_n^{+1},Z_n^{-1})$ to her final message $M_r$.
        \vspace{-1mm}
        \item Bob generates $A$ according to $\pi$.
    \end{enumerate}
    \vspace{0mm}
    \\
    \hline
    \end{tabular}
    \caption{The protocol $\pi^{<n}$ obtained by decomposing $\pi$.}
    \label{pi<n}
\end{figure}

\begin{figure}[H]
    \centering
    \begin{tabular}{|p{15.7cm}|}
    \hline ~\\
    \multicolumn{1}{|c|}{\textbf{A protocol $\pi^n$ over the input distribution $(X',Y')\sim\mu$}} \\
    \\ \hline
    \vspace{-2.85mm}
    \begin{enumerate}
        \item Alice pretends that her input consists of $n$ coordinates and writes it as $X=(X_1,\ldots,X_n)$.
        \hspace{-1mm}
        \item Bob pretends that his input consists of $n$ coordinates and writes it as $Y=(Y_1,\ldots,Y_n)$.
        \hspace{-1mm}
        \item Alice embeds $X'$ into $X_n$, and Bob embeds $Y'$ into $Y_n$.
        \hspace{-1mm}
        \item Using public randomness, the players draw $X_{<n}$ from its marginal under $\mu^{n-1}$. Bob then privately draws $Y_{<n}$ from its conditional distribution given $X_{<n}$.
        \hspace{-1mm}
        \item Alice and Bob simulate the $r$ rounds of $\pi$ on $(X,Y)$, generating the messages $M=(M_1,\ldots,M_r)$.
        \vspace{-1mm}
        \item Bob generates $A$ according to $\pi$.
    \end{enumerate}
    \vspace{0mm}
    \\
    \hline
    \end{tabular}
    \caption{The protocol $\pi^n$ obtained by decomposing $\pi$.}
    \label{pi_n}
\end{figure}

Equivalently, we may take the following distributional view of the two protocols:
\begin{itemize}
    \item $\pi^n=(X_n,Y_n,M,A,RX_{<n})$.
    \item $\pi^{<n}=(X_{<n},Y_{<n},(M,Z_n),A,RY_n)$, where $Z_n=(Z_n^{+1},Z_n^{-1})$ is computable by Alice alone. Here, $(M,Z_n)$ denotes the $r$-round transcript obtained by appending $Z_n$ to Alice's final message $M_r$.
\end{itemize}

For the analysis, let $\hat{\pi}^{<n}$ be the same protocol as $\pi^{<n}$, except that Alice sends $(\hat{Z}_n^{+1},\hat{Z}_n^{-1})$ instead of $(Z_n^{+1},Z_n^{-1})$. The following lemma follows from Claim 5.2 of \cite{SY26}.

\begin{lemma}[Claim 5.2 of \cite{SY26}]
$
\VarBobRound(\pi^n @ f)
+\VarObsRound(\hat{\pi}^{<n} @ f^{+(n-1)})
\leq \VarObsRound(\pi @ f^{+n}).
$
\label{clm:decomp_var_oneway}
\end{lemma}

We next bound the additional variance introduced by rounding.

\begin{claim} $
\VarObsRound(\pi^{<n} @ f^{+(n-1)})
\leq \VarObsRound(\hat{\pi}^{<n} @ f^{+(n-1)})+\delta^2.$\label{clm:decomp_var_oneway_v2}
\end{claim}

\begin{proof} 
Denote $ \widetilde{Z}_n^A := \Esub{Y_{<n}} \left[f^+(X_{<n}, Y_{<n}) \mid MRX_{<n}Y_n\hat{Z}_n^{A} A\right]$ the expectation conditioned on rounded value of $\hat{Z}_n^A$. Following the proof of \Cref{clm:decomp_var_oneway} (i.e. Claim 5.2 of SY26), we can derive that 

\begin{align}
    \VarObsRound(\pi^{<n} @ f^{+(n-1)}) & = \Esub{MRX_{<n}Y_nA} \ \Esub{Y_{<n}} \left[(f^+(X_{<n}Y_{<n}) - \widetilde{Z}_n^A)^2\mid MRX_{<n}Y_nA\right]
    \label{eqx}
\end{align}
Using \Cref{fct:var_bound}, for a fixed realization of $MRX_{<n}Y_nA$, we have 
\begin{align*}
    \Esub{Y_{<n}} \left[(f^+(X_{<n}Y_{<n}) - \widetilde{Z}_n^A)^2\mid MRX_{<n}Y_nA\right] & = \Varsub{Y_{<n}} \left[f^+(X_{<n}Y_{<n} \mid MRX_{<n}Y_nA\right] + (\widetilde{Z}_n^A - Z_n^A)^2 \\
    & \leq \Varsub{Y_{<n}} \left[f^+(X_{<n}Y_{<n} \mid MRX_{<n}Y_nA\right] + \delta^2 
\end{align*}
where the final inequality follows $|\widetilde{Z}_n^A - Z_n^A|\leq \delta$. Taking the expectation over $MRX_{<n}Y_nA$ and combining with (\ref{eqx}) completes the proof.
\end{proof}

We now prove the decomposition of variances.

\begin{claim}
$
\sum_{i=1}^n \VarBobRound(\pi^n @ f)
\leq \VarObsRound(\pi @ f^{+n})+ n\delta^2
$
\label{claim:var_decomp}
\end{claim}

\begin{proof}
    Combining \Cref{clm:decomp_var_oneway} and \Cref{clm:decomp_var_oneway_v2} gives $$
\VarBobRound(\pi^n @ f) + 
\VarObsRound(\pi^{<n} @ f^{+(n-1)}) \leq + \VarObsRound(\pi @ f^{+(n-1)}) + \delta^2.$$ Telescoping this inequality gives the desired bound.
\end{proof}

%% file: 5-variance-1way.tex
\section{Decomposition of Information Cost}
\label{sec:info}

In this section, we now show that the information cost also decomposes linearly, up to a ``constant'' additive loss. Here, we switch views to a direct construction of each protocol $\pi^i$, equivalent to the construction in \Cref{sec:var}.

\begin{figure}[H]
    \centering
    \begin{tabular}{|p{15.7cm}|}
    \hline ~\\
    \multicolumn{1}{|c|}{\textbf{A protocol $\pi^{i}$ over the input distribution $(X',Y') \sim \mu$}} \\ 
    \\ \hline
    \vspace{-2.85mm}
    \begin{enumerate}
        \item Alice pretends that her input consists of $n$ coordinates, and writes it as $X = (X_1,\ldots,X_n)$. 
        \hspace{-1mm}
        \item Bob pretends that his input consists of $n$ coordinates, and writes it as $Y = (Y_1,\ldots,Y_n)$.
        \hspace{-1mm}
        \item Alice embeds $X'$ into $X_i$, and Bob embeds $Y'$ into $Y_i$.
        \hspace{-1mm}
        \item Players use public randomness to jointly draw $X_{<i}Y_{> i}$.
        Then, Alice privately draws $X_{> i}$ conditioned on $Y_{>i}$, and Bob privately draws $Y_{<i}$ conditioned on $X_{<i}$. 
        \hspace{-1mm} 
        \item Alice and Bob simulate the $r$ rounds of communication $M = (M_1,\ldots,M_r)$ with respect to the protocol $\pi$ and the input $(X,Y)$.
        \vspace{-1mm}
        
        \item For each $j =i+1,...,n$, Alice computes $Z_j = (Z_j^{+1}, Z_j^{-1})$ where $$Z_j^{+1} := \Esub{Y_{<j}} \left[f^+(X_{<j}, Y_{<j}) \mid MRX_{<j}Y_{\geq j}, A=+1\right] \text{ rounded to the closest multiple of $\delta$,}$$ $$Z_j^{-1} := \Esub{Y_{<j}} \left[f^+(X_{<j}, Y_{<j}) \mid MRX_{<j}Y_{\geq j}, A=-1\right] \text{ rounded to the closest multiple of $\delta$,}$$ and sends them to Bob together with her last message.
        
        \item Bob generates $A$.
    \end{enumerate}
    \vspace{0mm}
    \\
    \hline
    \end{tabular}
    \caption{A direct construction of $\pi^i$.}
    \label{pi_n_new}
\end{figure}

Equivalently, we may view $\pi^i$ as a distribution $\pi^{i} = (X_{i}, Y_{i}, (M, Z_{>i}), A, RX_{<i}Y_{> i})$ which is precisely the distribution described in \Cref{sec:var}. As a crucial observation, we note that all $\pi^i$'s can be coupled from a single random process with the underlying distribution $\pi$.

Before analyzing the information costs, we first define helper variables and prove related bounds. For each $i \in [n]$ and $a \in \{\pm 1\}$, denote $D_i^a = Z_{i+1}^a - Z_i^a$, $H_i = MXY_{\geq i}$, and $S_i^a = \E \left[f^+(X, Y) \mid H_i, A=a\right]$.

\begin{claim} The following items hold.

\begin{enumerate}
    \item $H_{i+1} \subseteq H_i$
    \item $H_i$ determines $S_{\geq i}^{+1}, S_{\geq i}^{-1}, \hat{Z}_{\geq i}^{+1}, \hat{Z}_{\geq i}^{-1}, Z_{\geq i}^{+1}, Z_{\geq i}^{-1}$
    \item $(S_{n}^{+1},\ldots,S_1^{+1})$ is a (Doob) martingale under $\pi\mid A=+1$
    \item $(S_{n}^{-1},\ldots,S_1^{-1})$ is a (Doob) martingale under $\pi\mid A=-1$
\end{enumerate}
\label{claim:properties}
\end{claim}

\begin{proof} The first item is immediate from the description of $H_i$, and as a result, the item 3 and 4 are also immediate. It remains to justify item 2. $H_i$ determines $(S_i^{+1}, S_i^{-1})$ by definition. Using the rectangle property (\Cref{prop:rectangle}), we can write for each $a \in \{\pm 1\}$
    \begin{align*}
        \hat{Z}_i^a = \Esub{Y_{< i}} \left[f^+(X_{<i}, Y_{<i}) \mid MRX_{<i}Y_{\geq i}, A=a\right] & = \Esub{Y_{< i}} \left[f^+(X_{<i}, Y_{<i}) \mid MXY_{\geq i}, A=a\right] \\
        & = \Esub{Y_{< i}} \left[f^+(X_{<i}, Y_{<i}) \mid H_i, A=a\right]
    \end{align*}
which means a fixing of $H_i$ determines $(\hat{Z}_i^{+1},\hat{Z}_i^{-1})$, and therefore determines $(Z_i^{+1},Z_i^{-1})$ via rounding. Since $H_{>i} \subseteq H_i$, it also determines $S_{> i}^{+1}, S_{>i}^{-1}, \hat{Z}_{> i}^{+1}, \hat{Z}_{> i}^{-1}, Z_{> i}^{+1}, Z_{> i}^{-1}$
\end{proof}

\begin{claim}
 For each $a\in\{\pm1\}$, we have
$
\sum_{i=1}^{n-1}\Esub{\pi}[(S_{i+1}^a-S_i^a)^2]\leq1000n.
$
\label{claim:sum_diff_sq}
\end{claim}

\begin{proof}
We shall prove the statement for $a=+1$. Denote $\pi^+=\pi\mid A=+1$. Since $(S_{n}^{+1},\ldots,S_1^{+1})$ is a martingale under $\pi^+$, we have
\begin{align*}
\sum_{i=1}^{n-1}\Esub{\pi^+}[(S_{i+1}^{+1}-S_i^{+1})^2]
&=\Esub{\pi^+}[(S_1^{+1})^2]
-\Esub{\pi^+}[(S_n^{+1})^2] \leq\Esub{\pi^+}[f^+(X,Y)^2],
\end{align*}
where the last step uses $S_1^{+1}=f^+(X,Y)$.

Since $\sum_{i=1}^{n-1}(S_{i+1}^{+1}-S_i^{+1})^2$ is determined solely by $H_1 = MRXY$, we have
\begin{align*}
\pi(A=+1) \cdot \Esub{\pi^+}\left[\sum_{i=1}^{n-1}(S_{i+1}^{+1}-S_i^{+1})^2\right] 
&= \Esub{\pi}\left[\sum_{i=1}^{n-1}(S_{i+1}^{+1}-S_i^{+1})^2 \cdot \mathbbm{1}(A=+1)\right]\\
&= \Esub{\pi}\left[ \Esub{MRXY} \left[ \sum_{i=1}^{n-1}(S_{i+1}^{+1}-S_i^{+1})^2 \cdot \mathbbm{1}(A=+1) \mid MRXY\right]\right]\\
& \geq \Esub{\pi}\left[ 0.001 \cdot  \Esub{MRXY} \left[ \sum_{i=1}^{n-1}(S_{i+1}^{+1}-S_i^{+1})^2\right]\right]\\
& = 0.001 \cdot \Esub{\pi}\left[\sum_{i=1}^{n-1}(S_{i+1}^{+1}-S_i^{+1})^2\right] 
\end{align*}
where the inequality uses the fact that conditioned on $H_1 = MRXY$, the sum $\sum_{i=1}^{n-1}(S_{i+1}^{+1}-S_i^{+1})^2$ is determined and also the probability that $A=+1$ is at least 0.001 due to the smoothing assumption.
Combining the two bounds gives
\begin{align*}
\sum_{i=1}^{n-1}\Esub{\pi}[(S_{i+1}^{+1}-S_i^{+1})^2]
&\leq1000 \cdot \pi(A=+1) \cdot \Esub{\pi^+} \left[f^+(X,Y)^2\right]\\
&=1000\cdot \Esub{\pi}\left[f^+(X,Y)^2 \cdot \mathbbm{1}(A=+1)\right]\\
&\leq 1000\cdot \Esub{\pi}\left[f^+(X,Y)^2\right]\\
&\leq1000n.
\end{align*}
\end{proof}

\begin{claim}
For each $a\in\{\pm1\}$, we have
$
\sum_{i=1}^{n-1}\Esub{\pi}[(D_i^a)^2]=O(n).
$
\label{clm:sum_di_sq}
\end{claim}

\begin{proof}
Recall by the rectangle property (\Cref{prop:rectangle}) that     \begin{align*}
        \hat{Z}_i^a = \Esub{Y_{< i}} \left[f^+(X_{<i}, Y_{<i}) \mid MRX_{<i}Y_{\geq i}, A=a\right] & = \Esub{Y_{< i}} \left[f^+(X_{<i}, Y_{<i}) \mid MXY_{\geq i}, A=a\right] \\
        & = \Esub{Y_{< i}} \left[f^+(X_{<i}, Y_{<i}) \mid H_i, A=a\right]
    \end{align*}
and therefore $S_i^a=\hat Z_i^a+f^+(X_{\geq i},Y_{\geq i}).$
Consequently, we have $\hat Z_{i+1}^a-\hat Z_i^a
=(S_{i+1}^a-S_i^a)+f(X_i,Y_i).$ Recall also that $Z_i^a$ is $\hat{Z}_i^a$ rounded to the closest multiple of $\delta \in (0,1)$. Thus, we shall have
\begin{align*}
    D_i^a = Z_{i+1}^a - Z_{i}^a & = (S_{i+1}^a-S_i^a)+f(X_i,Y_i) + (Z_{i+1}^a-\hat{Z}_{i+1}^a) + (Z_{i}^a-\hat{Z}_{i}^a) \\
    & \leq (S_{i+1}^a-S_i^a) + 3
\end{align*}
which it follows that $(D_i^a)^2
\leq2(S_{i+1}^a-S_i^a)^2+18.$
Taking expectations and summing over $1\leq i<n$ with \Cref{claim:sum_diff_sq} completes the proof.
\end{proof}

We now prove the information cost decomposition. 

\begin{claim}
$
\sum_{i=1}^n\IC(\pi^i)
\leq\IC(\pi)+O(n/\delta^2)
$
\label{claim:IC_decomp}
\end{claim}

\begin{proof}
First, consider the LHS term from Bob's contribution. By the chain rule,
\begin{align*}
\sum_{i=1}^n I(MZ_{>i}:X_i\mid RX_{<i}Y_{\geq i})
&=\sum_{i=1}^n I(M:X_i\mid RX_{<i}Y_{\geq i})+\sum_{i=1}^n I(Z_{>i}:X_i\mid MRX_{<i}Y_{\geq i}).
\end{align*}
For the first RHS term, we bound:
\begin{align*}
    \sum_{i=1}^n  I(M: X_i \mid RX_{<k} Y_{\geq k}) & \leq \sum_{i=1}^n  I(MY_{<k}: X_i \mid RX_{<k} Y_{\geq k}) \\
    & = \sum_{i=1}^n  I(Y_{<k}: X_i \mid RX_{<k} Y_{\geq k}) + \sum_{i=1}^n  I(M: X_i \mid RX_{<k} Y) \tag{chain rule} \\
    & = I(M:X \mid RY) \tag{rectangle property and chain rule.}
\end{align*}
And for the second RHS term, we bound:
\begin{align*}
    \sum_{i=1}^n I(Z_{\geq i} : X_i \mid RX_{<i} Y_{\geq i}) & \leq  \sum_{i=1}^n I(ZY_{<i} : X_i \mid RX_{<i} Y_{\geq i}) \\
    & =  \sum_{i=1}^n I(Y_{<i} : X_i \mid RX_{<i} Y_{\geq i}) +  \sum_{i=1}^n I(Z : X_i \mid RX_{<i} Y) \tag{chain rule} \\
    & = I(Z:X \mid RY) \tag{rectangle property and chain rule}\\
    & \leq H(Z) \\
    & \leq H(Z^{+1}) + H(Z{-1}) \tag{$Z = (Z^{+1}, Z^{-1})$ and subadditivity}
\end{align*}

We will now bound $H(Z^{+1})$ and then the bound of $H(Z^{-1})$ follows analogously. Notice that $Z_1^{+1}=0$, and $Z^a$ and $(D_1^a,\ldots,D_{n-1}^a)$ determine each other. Moreover, $D_i^a/\delta$ is integer-valued by our rounding scheme. Applying \Cref{lem:entropy_bound}, we have
\begin{align*}
H(Z^{+1})
=H(D_1^{+1},\ldots,D_{n-1}^{+1}) \leq \sum_{i=1}^{n-1}H(D_i^{+1}) & = \sum_{i=1}^{n-1}H(D_i^{+1}/\delta)\\
&\leq c_0(n-1)+\frac{c_1}{\delta^2} \cdot 
\sum_{i=1}^{n-1}\Esub{\pi}[(D_i^a)^2]\\
&=O(n/\delta^2)
\end{align*}
where the final step uses \Cref{clm:sum_di_sq}. Hence $H(Z)=O(n/\delta^2)$.

Combining all the bounds above completes the proof.
\end{proof}

%% file: 6-discussion.tex
\section{Discussion}
\label{AI}

\paragraph{Comparison with \cite{SY26}.}
Our construction of $(\pi^1,\ldots,\pi^n)$ differs
from that of Sawettamalya and Yu only in the rounding scheme.
Specifically, we round each $\hat{Z}_i^a$ to the nearest multiple of
$\delta=\sqrt{\Delta}$, whereas they round $(\hat{Z}_i^a)^2$ to the nearest
multiple of $\Delta$. This change has two knock-on effects.
\begin{enumerate}
    \item Their rounding scheme makes the analogue of
    \Cref{clm:decomp_var_oneway_v2} immediate. Under our scheme,
    this bound requires a more involved argument, given in the
    proof of \Cref{clm:decomp_var_oneway_v2}.
    
    \item They bound the information cost incurred by $Z$ trivially by its length, obtaining $O(n\log(n/\Delta))$. We instead
    bound this cost by the \emph{entropy} of $Z$, which
    \Cref{claim:IC_decomp} bounds by $O(n/\Delta)$, and which is $O(n)$ for the choice of constant $\Delta$.
\end{enumerate}

It is also worth mentioning that the proof of \cite{SY26}
constructs the desired protocol $\eta$ through a
\emph{decomposition} process. A more direct and traditional
construction is through an \emph{embedding} approach, introduced by the seminal work of \cite{BBCR13}. Our proof
actually combines these two views, as presented in
\Cref{sec:var} and \Cref{sec:info}, respectively.

\paragraph{AI Usage.} At an early stage of this work, the author observed that the sequence $(S_n^a,\ldots,S_1^a)$ forms a martingale (\Cref{claim:properties}) whose variance is linearly bounded (\Cref{claim:sum_diff_sq}), which subsequently bound the entropy of $Z$. GPT 5.6 Sol helped formalize this observation by proposing an entropy bound for integer-valued random variables, which the author ultimately replaced with \Cref{lem:entropy_bound}. It also identified the need for the ``smoothing'' assumption, addressing a gap that the author had missed in the initial argument. The author takes full responsibility in the contents of this paper.

The manuscript was prepared with assistance from GPT 6.1 Sol in polishing the exposition and correcting minor errors.